\documentclass{easychair}

\usepackage[utf8]{inputenc}
\usepackage{microtype}
\usepackage{xspace}
\usepackage{amsmath}
\usepackage{amssymb}
\usepackage{cancel}
\usepackage{todonotes}
\usepackage{cite}

\usetikzlibrary{positioning}

\newcommand{\doi}[1]{\href{http://dx.doi.org/#1}{doi:\nolinkurl{#1}}}
\newcommand{\arxiv}[1]{\href{https://arxiv.org/abs/#1}{arXiv:\nolinkurl{#1}}}
\newcommand\toolname[1]{\mbox{\textsf{#1}}}
\newcommand\tool[1]{\toolname{#1}\xspace}
\newcommand\concon{\tool{ConCon}}
\newcommand\conconversion{1.5}
\newcommand\ceta{\tool{C\kern-0.2exe\kern-0.5exT\kern-0.5exA}}
\newcommand\csi{\tool{CSI}}
\newcommand\wm{\tool{Waldmeister}}
\newcommand\TRS{TRS\xspace}

\newcommand\CTRS{CTRS\xspace}
\newcommand\CTRSs{CTRSs\xspace}
\newcommand\mc[1]{\ensuremath{\mathcal{#1}}\xspace}
\newcommand\msf[1]{\ensuremath{\mathsf{#1}}\xspace}
\newcommand\FF{\mc{F}}
\newcommand\PP{\mc{P}}
\newcommand\RR{\mc{R}}
\newcommand\TT{\mc{T}}
\newcommand\VV{\mc{V}}
\newcommand\Pos{\PP\msf{os}}
\newcommand\Var{\ensuremath{\VV}\xspace}
\newcommand\RRu{\RR_{\msf{u}}}
\newcommand\rname{\rho}
\newcommand\IF{\Leftarrow}
\newcommand\EQ{\approx}
\newcommand\cond{\approx}
\newcommand\urule[2][]{\ell#1_{#2} \to r#1_{#2}}
\newcommand\crule[2][]{\urule[#1]{#2} \IF c#1_{#2}}
\newcommand\tcap[1][]{\ensuremath{\textsf{tcap}_{#1}}}
\newcommand\wrt{with respect to\xspace}
\newcommand\f[1]{\msf{#1}} 

\newcommand\lemref[1]{Lemma~\ref{lem:#1}}
\newcommand\propref[1]{Property~\ref{prop:#1}}
\newcommand\exref[1]{Example~\ref{ex:#1}}
\newcommand\figref[1]{Figure~\ref{fig:#1}}
\newcommand\cstepn[2][]{\xr[#1]{#2}}
\newcommand\cmd[1]{%
\begin{flushleft}
  \qquad \texttt{#1}
\end{flushleft}%
}

\makeatletter

\newcommand{\r@rrow}[3]{%
  \newcommand{#1}[2][]{%
    \def\next{#2\@ifempty{##1}{}{_{##1}}\@ifempty{##2}{}{^{##2}}}%
    \mathchoice{#3[##1]{##2}}{\next}{\next}{\next}%
  }%
}
\newcommand{\l@rrow}[3]{%
  \newcommand{#1}[2][]{%
    \def\next####1{%
      \setbox0=\hbox{$####1\vphantom{#2}\@ifempty{##1}{}{_{\vphantom{##1}}}%
      \@ifempty{##2}{}{^{##2}}$}%
      \setbox1=\hbox{$####1\vphantom{#2}\@ifempty{##1}{}{_{##1}}%
      \@ifempty{##2}{}{^{\vphantom{##2}}}$}%
      \setbox2=\vbox{\hbox to\wd0{}\hbox to\wd1{}}%
      \mathrel{\hskip\wd2\hskip-\wd0\box0\hskip-\wd1\box1{#2}}%
    }%
    \mathchoice{#3[##1]{##2}}{\next\textstyle}%
    {\next\scriptstyle}{\next\scriptscriptstyle}%
  }%
}

\l@rrow{\xl}{\leftarrow}{\xleftarrow}
\r@rrow{\xr}{\rightarrow}{\xrightarrow}
\newcommand{\xsr}{\rightsquigarrow}

\makeatother

\newcounter{counter}

\theoremstyle{definition}
\newtheorem{definition}[counter]{Definition}

\theoremstyle{plain}

\newtheorem{lemma}[counter]{Lemma}

\newtheorem{property}[counter]{Property}

\theoremstyle{remark}
\newtheorem{example}[counter]{Example}

\title{Certified Non-Confluence with {\concon~\conconversion}%
\thanks{This work is supported by FWF (Austrian Science Fund) project P27502.}}

\author{
Thomas Sternagel
\and
Christian Sternagel
}

\institute{
  University of Innsbruck,
  Austria\\
  \email{\{thomas,christian\}.sternagel@uibk.ac.at}
 }

\authorrunning{Sternagel and Sternagel}
\titlerunning{Certified Non-Confluence using {\concon~\conconversion}}

\begin{document}

\maketitle


%
We present three methods to check \CTRSs for non-confluence:
(1) an ad hoc method for 4-\CTRSs,
(2) a specialized method for unconditional critical  pairs,
and finally,
(3) a method that employs conditional narrowing to find non-confluence
witnesses.
%
We shortly describe our implementation
of these methods
in {\concon}~\cite{SM14}, 
then look into their certification with 
{\ceta}~\cite{TS09}, and finally conclude with
experiments on the confluence problems database (Cops).%
\footnote{%
  \url{http://cops.uibk.ac.at/?q=ctrs+oriented}
}

\section{Preliminaries}
\label{sec:prelim}

We assume familiarity with the basic notions of (conditional) term
rewriting~\cite{BN98,O02}, but shortly recapitulate terminology and notation
that we use in the remainder.
Given an arbitrary binary relation $\xr[]{}$, we write $\xl[]{}$, $\xr[]{+}$,
$\xr[]{*}$ for \emph{inverse}, \emph{transitive closure}, and \emph{reflexive
transitive closure}, respectively.
We use $\Var(\cdot)$ to denote the set of variables occurring in a given
syntactic object, like a term, a pair of terms, a list of terms, etc.
The set of terms $\TT(\FF,\VV)$ over a given signature of function symbols $\FF$
and set of variables $\VV$ is defined inductively:
$x \in \TT(\FF, \VV)$ for all variables $x \in \VV$,
and for every $n$-ary function symbol $f \in \FF$ and terms $t_1,\ldots,t_n \in
\TT(\FF,\VV)$ also $f(t_1,\ldots,t_n) \in \TT(\FF,\VV)$.
We say that terms $s$ and $t$ \emph{unify} if $s\sigma = t\sigma$ for some
substitution~$\sigma$.
The topmost part of a term that does not change under rewriting (sometimes
called its ``cap'') can be approximated for example by the $\tcap$
function~\cite{GTS05}. Informally, $\tcap(x)$ for a variable $x$ results in a
fresh variable, while $\tcap(t)$ for a non-variable term $t = f(t_1,\ldots,t_n)$
is obtained by recursively computing $u = f(\tcap(t_1),\ldots,\tcap(t_n))$ and
then asserting $\tcap(t) = u$ in case $u$ does not unify with any left-hand side
of rules in the \TRS $\RR$, and a fresh variable, otherwise.
We call a bijective variable substitution $\pi : \VV \to \VV$ a \emph{(variable)
renaming}.
A \emph{\CTRS} $\RR$ is a set of conditional rewrite rules of the shape
$\crule{}$ where $\ell$ and $r$ are terms and $c$ is a (possibly empty) sequence
of pairs of terms $s_1 \cond t_1, \ldots, s_k \cond t_k$.
For all rules in $\RR$ we have that $\ell \notin \VV$.
If additionally $\Var(r) \subseteq \Var(\ell,c)$ for all rules we call $\RR$ a
3-\CTRS otherwise a 4-\CTRS.
We restrict our attention to \emph{oriented} \CTRSs where conditions are
interpreted as reachability requirements.
The rewrite relation induced by an oriented \CTRS $\RR$ is structured into
\emph{levels}. For each level $i$, a \TRS $\RR_i$ is defined recursively as
follows:
$\RR_0 = \varnothing$, and
$\RR_{i+1} = \{\ell\sigma \to r\sigma \mid \crule{} \in \RR,
  \forall s \cond t \in c.\, s\sigma \cstepn[\RR_i]{*} t\sigma\}$.
The rewrite relation of $\RR$ is defined as
${\xr[\RR]{}} = {\bigcup_{i \geqslant 0}\xr[\RR_i]{}}$.
By dropping all conditions from a \CTRS~$\RR$ we obtain its \emph{underlying
\TRS}, denoted $\RRu$.
Note that ${\to_{\RR}}\subseteq{\to_{\RRu}}$.
We sometimes label rules like $\rname\colon \crule{}$.
Two variable-disjoint variants of rules $\crule{1}$ and $\crule{2}$ in
$\RR$ such that $\ell_1|_p \notin \VV$ and $\ell_1|_p\mu = \ell_2\mu$ with most
general unifier (mgu)~$\mu$, constitute a \emph{conditional overlap}.
A conditional overlap that does not result from overlapping two variants of the
same rule at the root, gives rise to a \emph{conditional critical pair} (CCP)
$\ell_1[r_2]_p\mu \EQ r_1\mu \IF c_1\mu,c_2\mu$.
A CCP $u \cond v \IF c$ is \emph{joinable} if
$u\sigma \xr[\RR]{*} \cdot \xl[\RR]{*} v\sigma$ for
all substitutions $\sigma$ such that $s\sigma \xr[\RR]{*} t\sigma$ for all $s
\cond t \in c$.
Moreover, a CCP $u \cond v \IF c$ is \emph{infeasible} if there is no
substitutions $\sigma$ such that $s\sigma \xr[\RR]{*} t\sigma$ for all $s \cond
t \in c$.
%

\section{Finding Witnesses for Non-Confluence of \CTRSs}
\label{sec:main}

To prove non-confluence of a \CTRS we have to find a witness, that is,
two diverging rewrite sequences starting at the
same term whose end points are not joinable.

The first criterion only works for \CTRSs that contain at least one
unconditional rule of type 4, that is, with extra-variables in the right-hand
side.

\begin{lemma}
  \label{lem:urnf}

  Given a 4-\CTRS $\RR$ and an unconditional rule $\rname\colon\urule{}$ in
  $\RR$ where $\Var(r) \nsubseteq \Var(\ell)$ and $r$ is a normal form \wrt
  $\RRu$ then $\RR$ is non-confluent.
\end{lemma}

\begin{proof}
  Since $\Var(r) \nsubseteq \Var(\ell)$ we can always find two renamings $\mu_1$
  and $\mu_2$ restricted to $\Var(r) \setminus \Var(\ell)$ such that $r\mu_1
  \xl[\rname]{} \ell\mu_1 = \ell\mu_2 \xr[\rname]{} r\mu_2$ and $r\mu_1 \neq
  r\mu_2$. As $r$ is a normal form \wrt $\RRu$ also $r\mu_1$ and $r\mu_2$ are
  (different) normal forms \wrt $\RRu$ (and hence also \wrt $\RR$).
  Because we found a non-joinable peak $\RR$ is non-confluent.
\end{proof}

\begin{example}
  Consider the second (unconditional) rule of the 4-\CTRS $\RR_{320} = \{
    \f{e} \to \f{f}(x) \IF \f{l} \cond \f{d},\
    \f{A} \to \f{h}(x, x)
  \}$
  from Cops.
  Its right-hand side $\f{h}(x,x)$ is a normal
  form \wrt the underlying \TRS of $\RR_{320}$ and the only variable occurring
  in it does not appear in its left-hand side $\f{A}$.
  So by \lemref{urnf} $\RR_{320}$ is non-confluent.
\end{example}

A natural candidate for diverging situations are the critical peaks of a \CTRS.
We will base our next criterion on the analysis of \emph{unconditional}
critical pairs (CPs) of \CTRSs. This restriction is necessary to guarantee the
existence of the actual peak. If we would also allow conditional CPs, we first
would have to check for infeasibility, since infeasibility is undecidable in
general these checks are potentially very costly (see for example~\cite{SM15}).

\begin{lemma}
  \label{lem:ucpnj}

  Given a \CTRS $\RR$ and an unconditional CP $s \EQ t$ of it. If $s$ and $t$
  are not joinable \wrt $\RRu$ then $\RR$ is non-confluent.
\end{lemma}

\begin{proof}

  The CP $s \EQ t$ originates from a critical overlap between two
  unconditional rules $\rname_1\colon\urule{1}$ and $\rname_2\colon\urule{2}$
  for some mgu $\mu$ of $\ell_1|_p$ and $\ell_2$ such that
  $s = \ell_1\mu[r_2\mu]_p \xl[]{} \ell_1\mu[\ell_2\mu]_p \xr[]{} r_1\mu = t$.
  Since $s$ and $t$ are not joinable \wrt $\RRu$ they are of course also not
  joinable \wrt $\RR$ and we have found a non-joinable peak.
  So $\RR$ is non-confluent.
\end{proof}

\begin{example}
  \label{ex:ucpnj}

  Consider the 3-\CTRS~$\RR_{271} = \{
    \f{p}(\f{q}(x)) \to \f{p}(\f{r}(x)),\
    \f{q}(\f{h}(x)) \to \f{r}(x),\
    \f{r}(x) \to \f{r}(\f{h}(x)) \IF \f{s}(x) \cond \f{0},\
    \f{s}(x) \to \f{1}
  \}$
  from Cops.
  First of all we can immediately drop the third rule because we can never
  satisfy its condition and so it does not influence the rewrite relation of
  $\RR_{271}$.
  This results in the \TRS $\RR_{271}'$.
  Now the left- and right-hand sides of the unconditional CP
  $\f{p}(\f{r}(z)) \EQ \f{p}(\f{r}(\f{h}(z)))$
  are not joinable because they are two different normal forms with respect to
  the underlying \TRS of $\RR_{271}'$.
  Hence $\RR_{271}$ is not confluent by \lemref{ucpnj}.
\end{example}

While the above lemmas are easy to check and we have fast methods to do so they
are also rather ad hoc.
A more general but potentially very expensive way to search for non-joinable
forks is to use conditional narrowing~\cite{MH94}.
%
%
\begin{definition}[Conditional narrowing]
  Given a \CTRS $\RR$ we say that $s$ \emph{(conditionally) narrows}
  to $t$, written $s \xsr_\sigma t$ if there is a variant of a rule
  $\rname\colon\crule{} \in \RR$,
  such that $\Var(s) \cap \Var(\rname) = \varnothing$ and
  $u \xsr_\sigma^* v$ for all $u \EQ v \in c$,
  a position $p\in\Pos_\FF(s)$,
  a unifier%
  \footnote{%
    In our implementation we
    start from an mgu of $s|_p$ and $\ell$ and
    extend it while trying to satisfy the conditions.
  }
  $\sigma$ of $s|_p$ and $\ell$,
  and $t = s[r]_p\sigma$.
  For a narrowing sequence
  $s_1 \xsr_{\sigma_1} s_2 \xsr_{\sigma_2} \cdots \xsr_{\sigma_{n-1}} s_n$
  of length $n$ we write $s_1 \xsr_\sigma^n s_n$
  where $\sigma = \sigma_1 \sigma_2 \cdots \sigma_{n-1}$.
  If we are not interested in the length we also write $s \xsr_{\sigma}^* t$.
\end{definition}
The following property of narrowing carries over from the unconditional case:
\begin{property}
  \label{prop:narrsound}

  If $s \xsr_\sigma t$ then $s\sigma \xr[]{} t\sigma$ with the same rule that
  was employed in the narrowing step.
  Moreover, if
  $s_1 \xsr_{\sigma_1} s_2 \xsr_{\sigma_2} \cdots \xsr_{\sigma_{n-1}} s_n$
  then
  $s_1 \sigma_1\sigma_2\cdots\sigma_{n-1} \xr[]{} s_2 \sigma_2\cdots\sigma_{n-1} \xr[]{}
  \cdots \xr[]{} s_n$.
  Again employing the same rule for each rewrite step as in the corresponding
  narrowing step.
\end{property}
Using conditional narrowing we can now formulate a more general non-confluence
criterion.
\begin{lemma}
  \label{lem:njf}

  Given a \CTRS $\RR$,
  if we can find two narrowing sequences $u \xsr_\sigma^* s$ and
  $v \xsr_\tau^* t$ such that $u\sigma\mu = v\tau\mu$ for some mgu $\mu$ and
  $s\sigma\mu$ and $t\tau\mu$ are not $\RRu$-joinable
  then $\RR$
  is non-confluent.
\end{lemma}
\begin{proof}
  Employing \propref{narrsound} we immediately get the two rewriting sequences
  $u\sigma \xr[\RR]{*} s\sigma$ and
  $v\tau \xr[\RR]{*} t\tau$.
  Since rewriting is closed under substitutions we have
  $s\sigma\mu \xl[\RR]{*} u\sigma\mu = v\tau\mu \xr[\RR]{*} t\tau\mu$.
  As the two endpoints of these forking sequences $s\sigma\mu$ and $t\tau\mu$
  are not joinable by $\RRu$ they are certainly also not joinable by $\RR$.
  This establishes non-confluence of the \CTRS $\RR$.
\end{proof}
\begin{example}
  \label{ex:njf}

  Consider the 3-\CTRS~$\RR_{262} = \{
    0 + y \to y,\
    \f{s}(x) + y \to x + \f{s}(y),\
    \f{f}(x, y) \to z \IF x + y \cond z + z'
  \}$
  from Cops.
  Starting from a variant of the left-hand side of the third rule
  $u = \f{f}(x',y')$
  we have a narrowing sequence
  $\f{f}(x',y') \xsr_\sigma x_1$
  using the variant
  $\f{f}(x_1,x_2) \to x_3 \IF x_1 + x_2 \cond x_3 + x_4$
  of the third rule and the substitution
  $\sigma = \{ x' \mapsto x_1, x_3 \mapsto x_1, x_4 \mapsto x_2 \}$.
  We also have another narrowing sequence
  $\f{f}(x',y') \xsr_\tau x_3$
  using the same variant of rule three and substitution
  $\tau = \{ x \mapsto x_3 + x_4, x' \mapsto \f{0}, y' \mapsto x_3 + x_4,
  x_1 \mapsto \f{0}, x_2 \mapsto x_3 + x_4 \}$
  where for the condition
  $x_1 + x_2 \cond x_3 + x_4$
  we have the narrowing sequence
  $x_1 + x_2 \xsr_{\tau} x_3 + x_4$,
  using a variant of the first rule
  $\f{0} + x \to x$.
  Finally, there is an mgu
  $\mu = \{ x_1 \mapsto \f{0}, x_2 \mapsto x_3 + x_4 \}$
  such that $u\sigma\mu = \f{f}(\f{0}, x_3 + x_4) = u\tau\mu$.
  Moreover,
  $x_1\sigma\mu = \f{0}$
  and
  $x_3\tau\mu = x_3$
  are two different normal forms.
  Hence $\RR_{262}$ is non-confluent by \lemref{njf}.
\end{example}

\section{Implementation}
\label{sec:impl}

Starting from its first participation in the confluence competition (CoCo)%
\footnote{%
  \url{http://coco.nue.riec.tohoku.ac.jp}
}
in 2014 {\concon~1.2.0.3} came equipped with some non-confluence heuristics.
Back then it only used Lemmas~\ref{lem:urnf} and \ref{lem:ucpnj}
and had no support for certification of the output.
In the next two years ({\concon~1.3.0} and~1.3.2) we focused on other
developments~\cite{SS15,SM15,SS16,SS_IWC16}
and nothing changed for the non-confluence part.
For this year's CoCo (2017) we have added \lemref{njf} employing conditional
narrowing to {\concon~\conconversion} and the output of all of the
non-confluence methods is now certifiable by \ceta.

Our implementation of \lemref{urnf} takes an unconditional rule
$\rname\colon\urule{}$, a substitution~$\sigma = \{ x \mapsto y\}$ with
$x\in\Var(r)\setminus\Var(\ell)$ and $y$ fresh w.r.t.\ $\rname$ and builds the
non-joinable fork $r \xl[\rname]{} \ell \xr[\rname]{} r\sigma$.

For \lemref{ucpnj} we have three concrete implementations that consider
an overlap from which an unconditional CP $s \EQ t$ arises:
The first of which just takes this overlap and then checks that $s$ and $t$ are
two different normal forms \wrt $\RRu$.
The second employs the {\tcap}-function to check for non-joinability, that is,
it checks whether $\tcap(s)$ and $\tcap(t)$ are not unifiable.
The third makes a special call to the \TRS confluence checker
{\csi}~\cite{ZFM11} providing the underlying \TRS $\RRu$ as well as the
unconditional CP $s \EQ t$ where all variables in $s$ and $t$ have been replaced
by fresh constants.
We issue the following command:
\cmd{csi -s '(nonconfluence -nonjoinability -steps 0 -tree)[30]' -C RT}
The strategy `\texttt{(nonconfluence -nonjoinability -steps 0 -tree)[30]}'
tells \csi to check non-joinability of two terms using tree automata techniques.
Here `\texttt{-steps 0}' means that {\csi} does not rewrite the input terms
further before checking non-joinability (this would be unsound in our setting).
The timeout is set to 30 seconds.
To encode the two terms for which we want to check non-joinability in the input
we set {\csi} to read relative-rewriting input (`\texttt{-C RT}').
We provide $\RRu$ in the usual Cops-format and add one line for the
CP $s \EQ t$ where its ``grounded'' left- and right-hand sides are related by
`\texttt{->=}', that is, we encode it as a relative-rule.
This is necessary to distinguish the unconditional CP from the rewrite rules.

Now, for an implementation of \lemref{njf} we have to be careful to respect the
freshness requirement of the variables in the used rule for every narrowing step
\wrt all the previous terms and rules.
The crucial point is to efficiently find the two narrowing sequences,
to this end we first restrict the set of terms from which to start narrowing.
As a heuristic we only consider the left-hand sides of rules of the \CTRS under
consideration.
Next we also prune the search space for narrowing.
Here we restrict the length of the narrowing sequences to at most three.
In experiments on Cops allowing sequences of length four or more did not yield
additional non-confluence proofs but slowed down the tool significantly to the
point where we lost other proofs.
Further, we also limit the recursion depth of conditional narrowing by
restricting the level (see the definition of the conditional rewrite relation in
the Preliminaries) to at most two.
Again, we set this limit as tradeoff after thorough experiments on Cops.
Finally, we use \propref{narrsound} to translate the forking narrowing sequences
into forking conditional rewriting sequences.
In this way we generate a lot of forking sequences so we only use fast methods,
like non-unifiability of the {\tcap}'s of the endpoints or that they are
different normal forms, to check for non-joinability of the endpoints.
Calls to {\csi} are to expensive in this context.


\section{Certification}
\label{sec:cert}

Certification
is quite similar for all of the
described methods.
We have to provide a non-confluence witness, that is, a non-joinable fork.
So besides the \CTRS $\RR$ under investigation we also need to provide the
starting term $s$, the two endpoints of the fork $t$ and $u$, as well as,
certificates for $s \xr[\RR]{+} t$ and $s \xr[\RR]{+} u$, and a certificate that
$t$ and $u$ are not joinable.
For the forking rewrite sequences we reuse a recent formalization of ours to
build the certificates (see~\cite{SS_CADE17}).
We also want to stress that because of \propref{narrsound} we did not have to
formalize conditional narrowing because going from narrowing to rewrite
sequences is already done in {\concon} and in the certificate only the rewrite
sequences show up.
For the non-joinability certificate of $t$ and $u$ there are three options:
either we state that $t$ and $u$ are two different normal forms or that
$\tcap(t)$ and $\tcap(u)$ are not unifiable; both of these checks are performed
within {\ceta}; or, when the witness was found by an external call to
{\csi}, we just include the generated non-joinability certificate.

\section{Experiments}
\label{sec:exp}

We tested the above non-confluence methods on all 129 oriented \CTRSs from Cops
(as of 2017-06-27).
Most of those are 3-\CTRSs but there are also four 4-\CTRSs.
On the whole {\concon~\conconversion} and {\ceta} are able to certify
non-confluence of 42 systems (including all 4-\CTRSs).
Last year's version of {\concon} can show non-confluence of 30 of the same 129
\CTRSs only using the implementation of Lemmas~\ref{lem:urnf} and
\ref{lem:ucpnj}.
By first removing infeasible rules, a feature which we have also recently
implemented in {\concon}, we gain another system (\texttt{271},
see~\exref{ucpnj}).
Another new feature is inlining of conditions (see~\cite{SS_CADE17}) which gives
two more systems (\texttt{351}, \texttt{353}).
Finally, with the help of conditional narrowing (\lemref{njf}) we gain another 9
systems (\texttt{272}, \texttt{328}, \texttt{330}, \texttt{352}, \texttt{391},
\texttt{404}, \texttt{410}, \texttt{411}, \texttt{524}).
In contrast {\concon~\conconversion} succeeds in showing confluence of 60
systems (where 7 use methods that are not certifiable yet, like using \wm to
show infeasibility of CCPs).
So from {\concon's} perspective only 27 of the 129 oriented \CTRSs of Cops are
still open.
These results are summarized in \figref{exp}.

Concerning future work we currently employ standard conditional narrowing in our
implementation. Implementing basic conditional narrowing or LSE conditional
narrowing should increase efficiency and maybe yield new \texttt{NO}-instances.

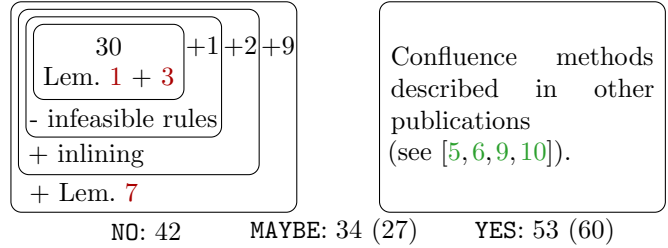
\begin{figure}
  \centering
  \begin{tikzpicture}
    \node {Lem.~\ref{lem:urnf} + \ref{lem:ucpnj}};
    \node at (0,0.4) {30};
    \node[rounded corners, shape=rectangle, draw, minimum width=20mm,
      minimum height=10mm] at (0,0.2) {};
    \node[rounded corners, shape=rectangle, draw, minimum width=26mm,
      minimum height=16mm] at (0.2,-0.0) {};
    \node at (0.2,-0.55) {- infeasible rules};
    \node at (1.25,0.4) {+1};
    \node[rounded corners, shape=rectangle, draw, minimum width=32mm,
      minimum height=22mm] at (0.4,-0.2) {};
    \node at (-0.3,-1.05) {+ inlining};
    \node at (1.75,0.4) {+2};
    \node[rounded corners, shape=rectangle, draw, minimum width=38mm,
      minimum height=28mm] at (0.6,-0.4) {};
    \node at (2.25,0.4) {+9};
    \node at (-0.33,-1.55) {+ Lem.~\ref{lem:njf}};
    \node[rounded corners, shape=rectangle, draw, minimum width=38mm,
      minimum height=28mm] at (5.5,-0.4) {};
    \node at (0.5,-2.05) {\texttt{NO}: 42};
    \node at (3.0,-2.05) {\texttt{MAYBE}: 34 (27)};
    \node at (5.8,-2.05) {\texttt{YES}: 53 (60)};
    \node at (5.5,-0.4) {\begin{minipage}{35mm}
      Confluence methods\linebreak
      described in other\linebreak
      publications (see~\cite{SS15,SM15,SS16,SS_IWC16}).
      \end{minipage}};
  \end{tikzpicture}
  \caption{(Non)-confluence results on 129 oriented \CTRSs.}
  \label{fig:exp}
\end{figure}

\paragraph{Acknowledgments.} We thank Bertram Felgenhauer for providing the
{\csi}-interface to check non-joinability of two terms. Also the first author
wants to thank Vincent van Oostrom for valuable insights during the
implementation of conditional narrowing.

\bibliographystyle{plainurl}
\bibliography{references}

\end{document}